\documentclass[10pt]{article}
\usepackage{amsmath, amsthm}
\usepackage{amssymb,amsmath}
\usepackage{psfrag,graphicx}

\newtheoremstyle{theorem}
  {10pt}          
  {10pt}  
  {\sl}  
  {\parindent}     
  {\bf}  
  {. }    
  { }    
  {}     
\theoremstyle{theorem}
\newtheorem{theorem}{Theorem}

\newtheorem{remark}[theorem]{Remark}
\newtheorem{proposition}[theorem]{Proposition}
\newtheorem{lemma}[theorem]{Lemma}

\newtheoremstyle{defi}
  {10pt}          
  {10pt}  
  {\rm}  
  {\parindent}     
  {\bf}  
  {. }    
  { }    
  {}     
\theoremstyle{defi}
\newtheorem{definition}[theorem]{Definition}

\def\J {\mathbb{J}}
\def\R {\mathbb{R}}
\def\N {\mathbb{N}}
\def\H {{\mathcal H}}

\def\B {{\mathcal B}}

\def\D {{\mathcal D}}
\def\E {{\mathcal E}}
\def\A {{\mathcal A}}
\def\L {{\mathcal L}}
\def\K {{\mathcal K}}
\def\J {{\mathcal J}}

\def\M {{\mathcal M}}

\def \l {\langle}
\def \r {\rangle}
\def \pt {\partial_t}
\def \ptt {\partial_{tt}}

\begin{document}

\title{\textbf{LONGTIME BEHAVIOR FOR OSCILLATIONS OF AN EXTENSIBLE
VISCOELASTIC BEAM WITH ELASTIC EXTERNAL SUPPLY}}


\author{Bochicchio Ivana$^1$ and Vuk Elena$^2$\\
$^1$Department of Mathematics and Informatics, \\Universit\`a di Salerno\\
Via Ponte don Melillo, 84084 Fisciano (SA), Italy\\
ibochicchio@unisa.it\\[2pt]
$^2$ Department of Mathematics, Universit\`a di Brescia\\
Via Valotti 9, 25133 Brescia, Italy\\
vuk@ing.unibs.it}

\maketitle \thispagestyle{empty}

\begin{abstract}
This work is focused on a nonlinear equation describing the
oscillations of an extensible viscoelastic beam with fixed ends,
subject to distributed elastic external force. For a general axial
load $\beta$, the existence of a finite/infinite set of stationary
solutions and buckling occurrence are scrutinized. The exponential
stability of the straight position is discussed. Finally, the
related dynamical system in the history space framework is shown
to possess a regular global attractor.


\end{abstract}

\section{Introduction}

In this paper we analyze the asymptotic behavior of the following
nonlinear dissipative evolution problem
\begin{equation}
\label{FIRST} \ptt u+\partial_{xxxx}u +\displaystyle\int_0^\infty
\mu(s) \partial_{xxxx}[u(t)-u(t-s)]ds- \big(\beta+\|\partial_x
u\|_{L^2(0,1)}^2\big)\partial_{xx}u= -ku + f,
\end{equation}
in the unknown variable $u=u(x,t):[0,1]\times\R\to\R$, which
describes the vibrations of an extensible viscoelastic beam of
unitary natural length. The real function $f=f(x)$ is the lateral
(static) load distribution, the term $-ku$ represents the lateral
action effected by the elastic foundation and the real parameter
$\beta$ denotes the axial force acting in the reference
configuration (positive in stretching, negative in compression).
The {\it memory kernel} $\mu$ is a nonnegative absolutely
continuous function on $\R^+=(0,\infty)$ (hence, differentiable
almost everywhere) such that
\begin{equation}
\label{MU} \mu'(s)+\delta\mu(s)\leq 0,\qquad
\int_0^\infty\mu(s)ds=\kappa,
\end{equation}
for some $\delta>0$ and $\kappa>0$. In particular, $\mu$ has an
exponential decay of rate $\delta$ at infinity.
\newline
Precisely, the model is obtained by combining the pioneering ideas
of Woinowsky-Krieger \cite{W} with the theory of Drozdov and
Kolmanovskii \cite{DK}, i.e. taking into account the geometric
nonlinearity due to the deflection, which produces an elongation
of the bar, and the energy loss due to the internal dissipation of
the material, which translates into a linear viscoelastic response
in bending.

In our analysis, we assume that both ends of the beam are hinged;
namely, for every $t\in\R$
\begin{equation}
\label{FIRSTBC}
u(0,t)=u(1,t)=\partial_{xx}u(0,t)=\partial_{xx}u(1,t)=0.
\end{equation}
Moreover, because of integro-differential nature of \eqref{FIRST},
the {\it past history} of $u$ (which need not fulfill the equation
for negative times) is assumed to be known. Hence, the initial
condition reads
\begin{equation}
\label{FIRSTIC} u(x,t)=u_\star(x,t),\qquad
(x,t)\in[0,1]\times(-\infty,0],
\end{equation}
where $u_\star:[0,1]\times (-\infty,0]\to\R$ is a given function.

In order to apply the theory of strongly continuous semigroups, we
recast the original problem as a differential system in the {\it
history space framework}. To this end, following
Dafermos~\cite{DAF}, we introduce the {\it relative displacement
history}
$$
\eta^t(x,s) = u(x,t)-u(x,t-s),
$$
so that equation (\ref{FIRST}) turns into
\begin{equation}
\label{BEAM}
\begin{cases}
\ptt u+\partial_{xxxx}u+\displaystyle\int_0^\infty \mu(s)
\partial_{xxxx}\eta(s)ds-
\big(\beta+\|\partial_x u\|_{L^2(0,1)}^2\big)\partial_{xx}u=  - ku + f,\\
\pt \eta=-\partial_s\eta+\pt u.
\end{cases}
\end{equation}
Accordingly, the initial condition \eqref{FIRSTIC} becomes
\begin{equation}
\label{BEAMIC}
\begin{cases}
u(x,0)=u_0(x), &  x\in[0,1],\\
\pt u(x,0)={u}_1(x), & x\in[0,1],\\
\eta^0(x,s)=\eta_0(x,s), & (x,s)\in [0,1]\times\R^+,
\end{cases}
\end{equation}
where we set
$$u_0(x)=u_\star(x,0),\qquad
u_1(x)=\pt u_\star(x,0),\qquad
\eta_0(x,s)=u_\star(x,0)-u_\star(x,-s).
$$
As far as the boundary conditions are concerned, \eqref{FIRSTBC},
for every $t\geq 0$, translates into
\begin{equation}
\label{BEAMBC}
\begin{cases}
u(0,t)=u(1,t)=\partial_{xx}u(0,t)=\partial_{xx}u(1,t)=0,\\
\eta^t(0,s)=\eta^t(1,s)=\partial_{xx}\eta^t(0,s)=\partial_{xx}\eta^t(1,s)=0,\\
\eta^t(x,0)= \displaystyle\lim_{s\to 0}{\eta}^t(x,s) = 0.
\end{cases}
\end{equation}

It is worth noting that the static counterpart of
problem~\eqref{FIRST} reduces to
\begin{eqnarray}
\label{STATIC1}
\begin{cases}
 u{''''}
-\left(\beta+\|\partial_x u\|_{L^2(0,1)}^2\right)u{''}+ku=f,\\
u(0)=u(1)=u{''}(0)=u{''}(1)=0.
\end{cases}
\end{eqnarray}
When $k\equiv0$ the investigation of the solutions to
\eqref{STATIC1} and their stability, in dependence on $\beta$,
represents a classical {\it nonlinear buckling problem}
(see, for instance, \cite{B1,DIC,RM}). Its numerical solutions
are available in  \cite{CC}, while their connection with
industrial applications is discussed in \cite{FW}. Recently, a
careful analysis of the corresponding buckled stationary states
was performed in \cite{NE} for all values of $\beta$. In
\cite{CZGP} this analysis was improved to include a source $f$
with a general shape.

For every $k>0$ and vanishing sources, exact solutions to
\eqref{STATIC1} can be found in \cite{BOV}, and at a first sight,
this case looks like a slight modification of previously
scrutinized models where $k$ vanishes. This is partially true.
Indeed, the restoring elastic force, acting on each point of the
beam, opposes the buckling phenomenon and increases the critical
Euler buckling value $\beta_c$, which is no longer equal to
$\sqrt{\lambda_1}$, the root square of the first eingenvalue of
the $\partial_{xxxx}$ operator, but turns out to be a {\it
piecewise} linear function of $k$. When the lateral load $f$
vanishes, the null solution is unique provided that
$\beta\geq-\beta_c(k)$ (see Theorem \ref{TH-stat.solut}), and
buckles when $\beta$ exceeds this critical value, as well as in
the case $k=0$. On the contrary, for some special positive values
of $k$, called {\it resonant values}, infinitely many solutions
may occur.

Moreover, in the case $k\equiv0$, if $f$ vanishes, the exponential
decay of the energy is provided when $\beta>-\lambda_1$, so that
the unique null solution is exponentially stable. On the contrary,
as the axial load $\beta\leq-\lambda_1$ the straight position
loses stability and the beam buckles. So, a finite number of
buckled solutions occurs and the global attractor coincides with
the unstable trajectories connecting them.

 By paralleling the results for $k=0$, the null solution is expected to be exponentially stable, when it is unique. Quite surprisingly, it is not so. For large values of $k$, the energy decays with a sub-exponential rate when $-\bar\beta>\beta>-\beta_c$ (see Theorem \ref{exp-stab}). In particular, for any fixed $k>\lambda_1$, the positive limiting value $\bar\beta(k)$ is smaller than the critical value $\beta_c(k)$, and the first overlaps the latter only if $0\leq k\leq \lambda_1$.

The motion equation (\ref{FIRST}) with $\mu=k=f=0$ turns out to be
conservative and has been considered for hinged ends in
\cite{B1,DI1}, with particular reference to well-posedness
results. Adding an external viscous damping term $\delta\pt u$
($\delta>0$) to this conservative model, stability properties of
the unbuckled (trivial) states have been established in
\cite{B,DIC} and, more formally, in \cite{RM}. In this case, the
global dynamics of solutions for a general $\beta$ has been
tackled first in \cite{HAL}, where some regularity of the
attractor is obtained provided that $\delta$ is large enough. When
the stiffness of the surrounding medium is neglected ($k=0$), the
existence of a regular attractor was proved for extensible
Kirchhoff beam \cite{CZ}, extensible viscoelastic \cite{GPV} and
thermo-elastic \cite{GNPP} beams. A similar result for an
extensible elastic beam resting on a viscoelastic foundation was
obtained in \cite {BOV}. This strategy can be generalized to the
investigation of nonlinear dissipative models which describe the
vibrations of extensible thermoelastic beams \cite{BGV}, and to
the analysis of the long term damped dynamics of extensible
elastic bridges suspended by flexible and elastic cables
\cite{BGV1}. In the last case the term $-ku$ is replaced by
$-ku^+$ and it represents a restoring force due to the cables.
Moreover, our approach may be adapted to the study of simply
supported bridges subjected to moving vertical load
\cite{venuti1}.

The final result of this work concern the existence of a regular
global attractor for all values of the real parameter $\beta$. The
main difficulty comes from the very weak dissipation exhibited by
the model, entirely contributed by the memory term. So, the
existence of the global attractor is stated through the existence
of a Lyapunov functional and the asymptotic smoothing property of
the semigroup generated by the abstract problem via a suitable
decomposition first proposed in \cite{GPV}.

\section{The Abstract Setting}

\noindent

In this section we will consider an abstract version of
problem~\eqref{BEAM}-\eqref{BEAMBC}. To this aim, let
$(H_0,\l\cdot,\cdot\r,\|\cdot\|)$ be a real Hilbert space, and let
$A:\D(A)\Subset H_0\to H_0$ be a strictly positive selfadjoint
operator. For $\ell\in\R$, we introduce the scale of Hilbert
spaces
$$
H_\ell=\D(A^{\ell/4}),\qquad \l u,v\r_\ell=\l
A^{\ell/4}u,A^{\ell/4}v\r,\qquad \|u\|_\ell=\|A^{\ell/4}u\|.
$$
In particular, $H_{\ell+1}\Subset H_\ell$ and the generalized
Poincar\'e inequalities hold
\begin{equation}\label{POINCARE}
\sqrt{\lambda_1}\,\|u\|_\ell^2\leq \|u\|_{\ell+1}^2, \quad \forall
u \in H_{\ell+1},
\end{equation}
where $\lambda_1>0$ is the first eigenvalue of $A$.

Given $\mu$ satisfying \eqref{MU}, we consider the $L^2$-weighted
spaces
$$
\M_\ell=L^2_\mu(\R^+,H_{\ell+2}), \,\,\, \l \eta,\xi\r_{\ell,\mu}=
\int_0^\infty\mu(s)\l\eta(s),\xi(s)\r_{\ell+2} ds, \,\,\,
\|\eta\|_{\ell,\mu}^2=\l \eta,\eta\r_{\ell,\mu}
$$
along with the infinitesimal generator of the right-translation
semigroup on $\M_0$, that is, the linear operator
$$T\eta=-D\eta,\qquad \D(T)=\{\eta\in{\M_0}:D\eta\in\M_0,\,\,\eta(0)=0\},$$
where $D$ stands for the distributional derivative, and
$\eta(0)=\lim_{s\to 0}\eta(s)$ in $H_2$.

Besides, we denote by $\M^1_\ell$ the weighted Sobolev spaces
$$\M^1_\ell=H^1_\mu(\R^+,H_{\ell+2})=\{\eta\in{\M_\ell}:D\eta\in\M_\ell\}, \qquad  \|\eta\|_{\M^1_\ell}^2=\|\eta\|_{\ell,\mu}^2
+\|D\eta\|_{\ell,\mu}^2.$$ Moreover, the functional
$$\J(\eta)=-\int_0^\infty\mu'(s)\|\eta(s)\|_{2}^2ds
$$
is finite provided that $\eta\in\D(T)$. From the assumption
\eqref{MU} on $\mu$,
\begin{equation}
\label{J} \|\eta\|_{0,\mu}^2\leq\frac{1}{\delta}\J(\eta).
\end{equation}
Finally, we define the product Hilbert spaces
$$\H_\ell=H_{\ell+2}\times H_\ell\times\M_{\ell}.
$$
For $\beta\in\R$ and $f\in H_{0}$, we investigate the evolution
system on $\H_0$ in the unknowns $u(t):[0,\infty)\to H_2$, $\pt
u(t):[0,\infty)\to H_0$ and $\eta^t:[0,\infty)\to\M_0$
\begin{equation}
\label{BASE}
\begin{cases}
\ptt u+Au+\displaystyle\int_0^\infty \mu(s) A\eta(s)ds+
\big(\beta+\|u\|^2_1\big)A^{1/2}u= -ku + f,\\
\pt \eta=T\eta+\pt u,
\end{cases}
\end{equation}
with initial conditions
$$
(u(0),u_t(0),\eta^0)=(u_0,u_1,\eta_0)=z\in\H_0.
$$

\begin{remark}
Problem~\eqref{BEAM}-\eqref{BEAMBC} is just a particular case of
the abstract system \eqref{BASE}, obtained by setting
$H_0=L^2(0,1)$ and
$$A=\partial_{xxxx}, \qquad \D(\partial_{xxxx})=\{w\in H^4(0,1) : w(0)=w(1)=w''(0)=w''(1)=0\}.$$
\end{remark}
\noindent This operator is strictly positive selfadjoint with
compact inverse, and its discrete spectrum is given by
$\lambda_n=n^4\pi^4$, $n\in\N$. Thus, ${\lambda_1=\pi^4}$ is the
smallest eigenvalue. Besides, the peculiar relation $
(\partial_{xxxx})^{1/2}=-\partial_{xx} $ holds true, with
$\D(-\partial_{xx})=H^2(0,1)\cap H^1_0(0,1).$

Besides, system~\eqref{BASE} generates a strongly continuous
semigroup (or dynamical system) $S(t)$ on $\H_0$ which
continuously depends on the initial data: for any initial data
$z\in\H_0$, $S(t)z$ is the unique weak solution to \eqref{BASE},
with related (twice the) energy given by
$$\E(t)=\|S(t)z\|^2_{\H_0}=\|u(t)\|^2_2+\|\pt u(t)\|^2+\|\eta^t\|^2_{0,\mu}.$$
We omit the proof of these facts, which can be demonstrated either
by means of a Galerkin procedure or with a standard fixed point
method. In both cases, it is crucial to have uniform energy
estimates on any finite time-interval.


\section {Steady states}

\noindent In the concrete problem \eqref{BEAM}-\eqref{BEAMBC},
taking for simplicity $f=0$, the stationary solutions $(u,0,0)$
solve the boundary value problem
\begin{equation}
\begin{cases}
\label{STATIC}
u''''-\big(\beta+\|u'\|^2_{L^2(0,1)}\big)u'' + ku= 0,\\
\noalign{\vskip1mm} u(0)=u(1)=u''(0)=u''(1)=0.
\end{cases}
\end{equation}
Our aim is to analyze the multiplicity of such solutions.

For every $k>0$,  let
$$
\mu_n(k)=\frac{k}{n^2\pi^2}+ n^2\,\pi^2\,,\qquad \beta_c(k)=
\min_{n\in\N}\mu_n(k).
$$
Assuming that $n_k\in\N$ be such that
$\displaystyle\mu_{n_k}=\min_{n\in\N}\mu_n(k)$, then it satisfies
$$
(n_k-1)^2n_k^2\leq \frac k {\pi^4}< n_k^2(n_k+1)^2\,.
$$
As a consequence, $\beta_c(k)$ is a piecewise-linear function of
$k$.

We consider the {\it resonant set}
$${\mathcal R}=\{ i^2j^2\pi^4:i,j\in\N, i<j\}.$$
When $k\in{\mathcal R}$ there exists at least a value   $\mu_j(k)$
which is not simple and
 $\mu_i=\mu_j$, $i\neq j$, provided that $k= i^2j^2\pi^4$ (resonant values). In the sequel, let $\mu_m(k)$ be the smallest value of $\{\mu_n\}_{n\in\N}$ which is not simple. Of course, the $\mu_n(k)$ are all simple and increasingly ordered with $n$ whenever $k<4\pi^4$.  Given $k>0$, let $n_\star$ be the integer-valued function given by
 \begin{equation}\nonumber
n_{\star }(\beta)=|{\mathcal N}_\beta|\,, \qquad{\mathcal
N}_\beta= \{n\,\in \mathbb{N}:\beta +\mu_n(k)<0\},
\end{equation}
where $|{\mathcal N}|$ stands for the cardinality of the set
${\mathcal N}$.

In the homogeneous case, we are able to establish the exact number
of stationary solutions and their explicit form. In particular, we
show that there is always at least one solution, and at most a
finite number of solutions, whenever the values of $\mu_n(k)$ not
exceeding $-\beta$ are simple.

\begin{theorem}{\rm{(see \cite{BOV})}}
\label{TH-stat.solut}
 If $\beta \geq - \beta_{c}(k)$, then for every $k>0$ system \eqref{STATIC} has only the null solution, depicting the straight equilibrium position. Otherwise:
\begin{itemize}
\item if $k\in{\mathcal R}$ and $\beta < -\mu_m(k)$, the smallest
non simple eigenvalue, there are infinitely many solutions; \item
if either $k\in{\mathcal R}$ and $-\mu_m(k)\leq \beta < -
\beta_{c}(k)$, or $k\not\in{\mathcal R}$ and $\beta < -
\beta_{c}(k)$, then besides the null solution there are also
$2n_\star (\beta)$ buckled solutions, namely
\begin{equation}
\label{soluz.tutte} u_{n}^{\pm }(x)\,=\,A_{n}^{\pm }\,\sin (n\pi
x)\,,\quad\,n=1,2,\ldots ,n_{\star }
\end{equation}
with
\begin{equation}
\label{staticresponse} A_{n}^{\pm }=\pm \,\frac{1}{n\,\pi
}\sqrt{-\,2\,\left[ \beta + \mu_n(k)\right]}.
\end{equation}
\end{itemize}
\end{theorem}

\begin{remark}
Assuming $k=0$ we recover the results of \cite{CZGP}.
\end{remark}

\noindent When  $k\not\in{\mathcal R}$ the set of all stationary
states is finite. Depending on the value of $k$ and $\beta$, the
solutions branch in the pairs from the unbuckled state
$A^\pm_n\,=\,0$ at the critical value $\beta =  - \beta_{c}(k)$,
i.e., the beam can buckle in either the positive or negative
directions of the transverse displacement. These branches exist
for all $\beta <  - \beta_{c}(k)$ and $A^\pm_n$ are monotone
increasing functions of $|\beta|$. For each $n$,
\eqref{staticresponse} admits real (buckled) solutions $A_{n}^{\pm
}$ if and only if $\beta < -\mu_n$. When $k<4\pi^4$, for any
$\beta$ in the interval
$$ - \frac{k}{(n+1)^2\,\pi^2}-(n+1)^2\,\pi^2<\beta < - \frac{k}{n^2\,\pi^2}-n^2\,\pi^2$$
the set ${\mathcal S_0}$ of the stationary solutions contains
exactly $2n_{\star}+1$ stationary points: the null solution and
the solutions represented by \eqref{soluz.tutte}.

When  $k\in{\mathcal R}$ the set ${\mathcal S_0}$ contains an
infinite numbers of solutions and all the resonant values are
obtained by solving
$$
n^2\pi^2 + \frac{k}{n^2\pi^2}= m^2\pi^2 +
\frac{k}{m^2\pi^2}\,,\qquad m,n\in\N\,, \ n>m\,.
$$
The smallest value $k$ is then equal to $4\pi^4$ and occurs when
$n=2$, $m=1$. In the sequel we present a sketch of different
bifurcation pictures occurring when $k=\pi^4$ (see Fig.1), and
$k=9\pi^4$ (see Fig.2).

\medskip
\begin{center}
\includegraphics[width=12cm, height=6cm]{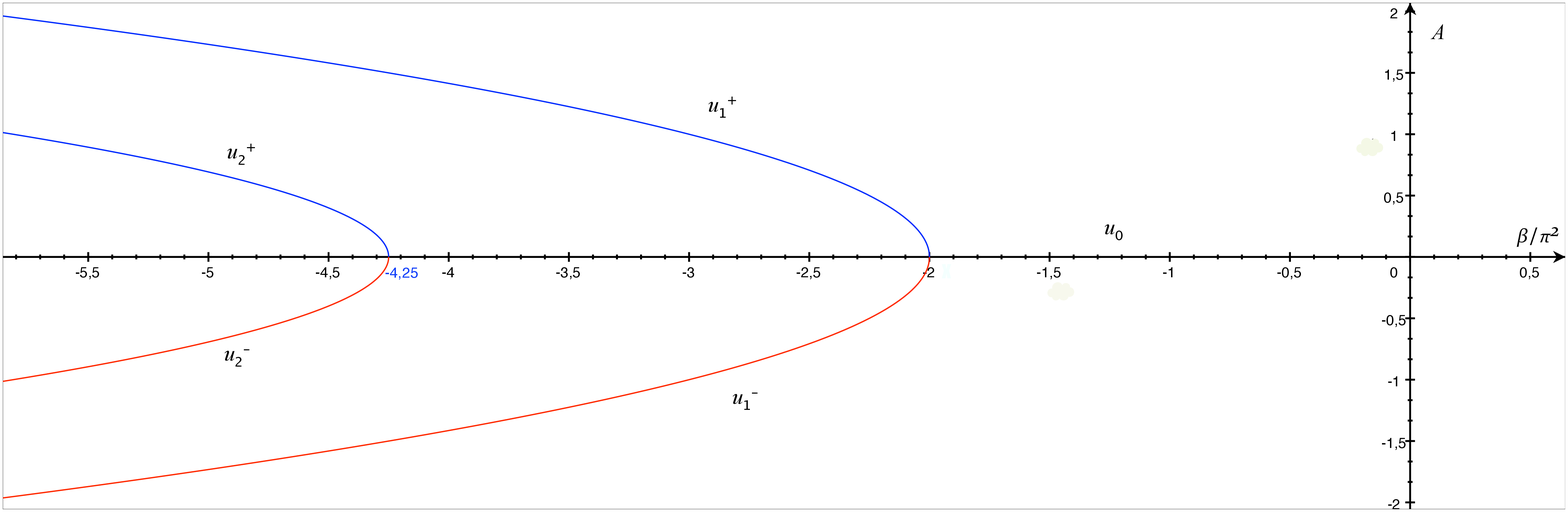}
\end{center}
 \begin{center}
 {\small {Fig.1. The bifurcation picture when $k=\pi^4$}}
\medskip
  \end{center}

\medskip
\begin{center}
\includegraphics[width=12cm, height=6cm]{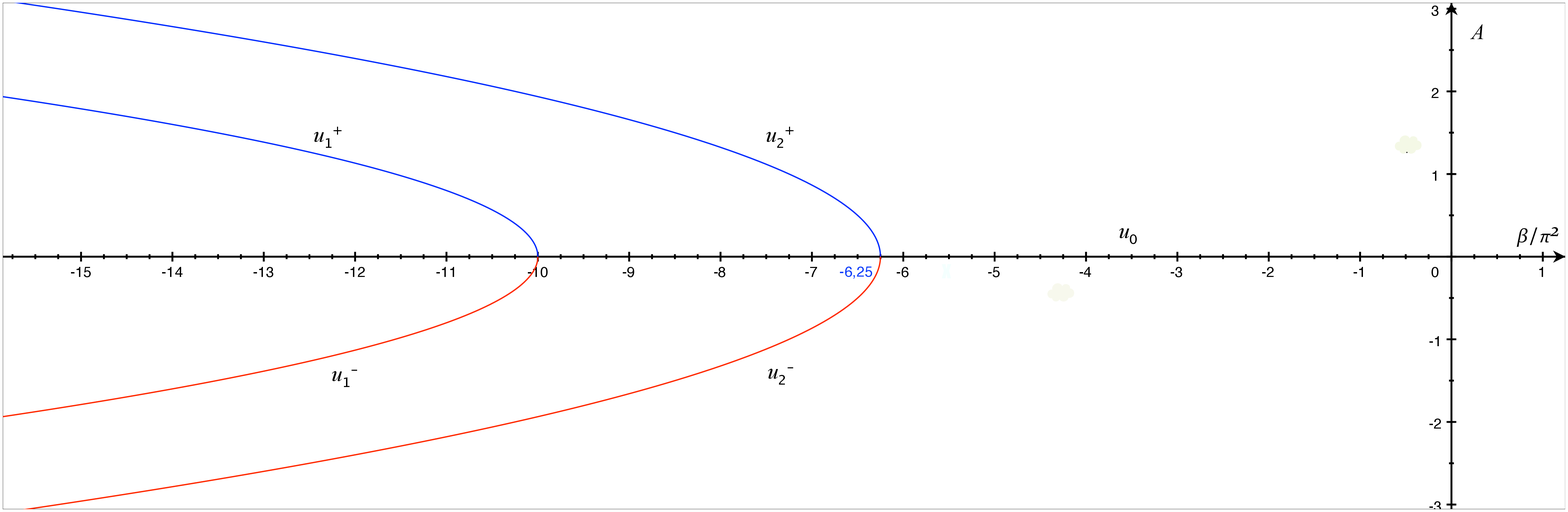}

\end{center}
 \begin{center}
 {\small {Fig.2. The bifurcation picture when $k=9\pi^4$}}
 \medskip
  \end{center}


\section {The Lyapunov functional}

\noindent It is well known that the absorbing set gives a first
rough estimate of the dissipativity of the system. In addition, it
is the preliminary step to scrutinize its asymptotic dynamics and
hence to prove the existence of a global attractor. Unfortunately,
when the dissipation is very weak, a direct proof via explicit
energy estimates might be very hard to find. For a quite general
class of the so-called {\it gradient systems} it is possible to
use an alternative approach appealing to the existence of a
Lyapunov functional. This technique has been adopted in
\cite{GPV}.
\begin {definition}
The Lyapunov functional is a function $\L\in C(\H_0,\R)$
satisfying the following conditions:
\begin{itemize}
\item[(i)] $\L(z)\to+\infty$ if and only if
$\|z\|_{\H_0}\to+\infty$; \item[(ii)] $\L(S(t)z)$ is nonincreasing
for any $z\in\H_0$; \item[(iii)] $\L(S(t)z)=\L(z)$ for all $t>0$
implies that $z\in{\mathcal S}$.
\end{itemize}
\end {definition}

\begin{proposition}
\label{LYAP} The function
$$\L(t)\,=\,\E(t)+\frac{1}{2}\left( \beta +\left\| u(t)\right\|
_{1}^{2}\right) ^{2}+k\,\left\| u(t)\right\| ^{2}-2\left\langle
f,\,u(t)\right\rangle\,
$$ is a Lyapunov functional for $S(t)$.
\end{proposition}

\begin{proof}
The continuity of $\L$ and assertion (i) above are clear. Using
\eqref{BASE}, we obtain quite directly the inequality
\begin{equation}\label{LYAP_IN}
\frac{d}{dt}\L(S(t)z)\leq-\delta\|\eta^t\|^2_{0,\mu},
\end{equation}
which proves the monotonicity of $\L$ along the trajectories
departing from $z$. Finally, if $\L(S(t)z)$ is constant in time,
we have that $\eta^t=0$ for all $t$, which implies that $u(t)$ is
constant. Hence, $z=S(t)z=(u_0,0,0)$ for all $t$, that is,
$z\in{\mathcal S}$.
\end{proof}

The existence of a Lyapunov functional ensures that

\begin {lemma}
\label {lemma nuovo} For all $t>0$ and initial data $z\in \H_0$,
with $\|z\|_{\H_0}\leq R$, there exists a positive constant $C$
(depending on $\left\| f\right\|$ and $R$) such that
\label{energy}
$$\E(t)\leq C.$$
\end{lemma}

\begin{proof}
Inequality \eqref{LYAP_IN} ensures that
\begin{equation} \nonumber
\begin{split}
\L(t) \leq \L(0)\leq C_0(R, \|f\|) \ .
\end{split}
\end{equation}
Moreover, taking into account that
$$\left\| u(t)\right\| ^{2} \leq \frac{1}{\lambda_1}\left\| u(t)\right\|_{2} ^{2} \leq \frac{1}{ \lambda_1}\E(t) = C_1\E(t)$$
we obtain the estimate
$$\L(t) \geq \E(t) - 2\left\langle f\,,\,u(t)\right\rangle \geq \E(t) - \frac{1}{\varepsilon} \left\| f\right\| ^{2} - \varepsilon\left\| u(t)\right\| ^{2} \geq
(1 - \varepsilon C_1)\E(t) - \frac{1}{\varepsilon} \left\|
f\right\| ^{2}.$$ Finally, fixing $\varepsilon < \frac{1}{C_1}$,
we have
$$\E(t)\leq \frac{1}{1 - \varepsilon C_1}(\L(0) +
 \frac{1}{\varepsilon} \left\| f\right\| ^{2}) \leq \frac{1}{1 - \varepsilon C_1}\left(C_0(R, \|f\|) +
  \frac{1}{\varepsilon} \left\| f\right\| ^{2}\right) = C.$$
\end{proof}

Moreover, defined the following functional \label{PHI_FUNZIONALE}
$$\Phi (t)\,=\,E(t)+\varepsilon \left\langle
\partial_t u\,,\,u\right\rangle,
$$
where $E(t)\,=\,\E(t)+\frac{1}{2}\left( \beta +\left\|
u(t)\right\| _{1}^{2}\right) ^{2}+k\,\left\| u(t)\right\| ^{2}$,
we can prove
\begin{lemma}
\label{lemma_34} For any given $z \in \mathcal{H}_0$ and for any
$t>0$ and $\beta \in \mathbb{R}$, when $\varepsilon$ is small
enough there exist three positive constants, $m_0 $, $m_1 $  and
$m_2$, independent of $t$ such that
$$m_{0}\,\mathcal{E}(t)\leq \Phi(t)\leq m_{1}\,\mathcal{E}(t)\,+\,m_2.$$
\end{lemma}

\begin{proof}
In order to prove the lower inequality we must observe that, by
Young inequality
$$
\left| \left\langle \partial_t u\,,\,u\right\rangle \right| \geq
-\frac{1}{2}\,\left\| \partial_t u\right\| ^{2}-\frac{1}{2}\left\|
u\right\| ^{2} \ ;
$$
hence, we obtain
\begin{equation} \nonumber
\begin{split}
\Phi (t)\, \geq &
\,\|u(t)\|^2_2+\left(1-\frac{\varepsilon}{2}\right)\|\pt
u(t)\|^2+\frac{1}{2}\left( \beta +\left\| u(t)\right\|
_{1}^{2}\right) ^{2}+
\\
& +\left(k\,-\frac{\varepsilon}{2}\right)\,\left\| u(t)\right\|
^{2} +\|\eta^t\|^2_{0,\mu}\ .
\end{split}
\end{equation}
If we choose $\varepsilon$ small enough to satisfy $\varepsilon<2$
and $\varepsilon<2k$, then we have
\begin{equation} \label{stima}
\Phi (t)\, \geq m_0 \, E(t) \geq m_0 \, \E(t),
\end{equation}
where $m_0\,=\,\min\{1- \frac{\varepsilon}{2}, 1-
\frac{\varepsilon}{2k}\}$.

The upper inequality can be obtained using the definition of
$\Phi$ and applying the estimate
\begin{equation} \label{aggiunta}
\left| \left\langle \partial_t u\,,\,u\right\rangle \right| \leq
\frac{1}{2}\,\left\| \partial_t u\right\|
^{2}+\frac{1}{2\lambda_1}\left\| u\right\|_2^{2} \ .
\end{equation}
First, we can write
\begin{equation} \nonumber
\begin{split}
\Phi (t) \leq &
 \left[ 1+\frac{1}{ \,\lambda _{1}}\left(
k+\frac{\varepsilon }{2}\right) \right] \left\| u(t)\right\|
_{2}^{2}+\left( \frac{\varepsilon }{2}+1\right) \left\|
\partial_t u(t)\right\| ^{2}+
\\
& + \frac{1}{2}\left( \beta +\left\| u(t)\right\| _{1}^{2}\right)
^{2}+\|\eta^t\|^2_{0,\mu}.
\end{split}
\end{equation}
Then, by (\ref{POINCARE}) and Lemma \ref{lemma nuovo} we infer
\begin{equation} \label{stima2}
\left( \beta +\left\|u\right\|_1 ^{2}\right)  \leq \left| \beta
\right| +\frac{1}{ \sqrt{\lambda _{1}} }\,C = \bar {C},
\end{equation}
so that we finally obtain
\begin{equation} \nonumber
\begin{split}
\Phi (t) &\leq \left[ 2+\frac{1}{ \,\lambda _{1}}\left(
k+\frac{\varepsilon }{2}\right) +\frac{\varepsilon }{2}\right]
\E(t) +\frac{1}{2}\bar {C}^2 \, = \, m_1\,\E(t) + m_2 \ .
\end{split}
\end{equation}
\end{proof}

\noindent As a byproduct, we deduce the existence of a bounded
absorbing set $\B_0$, chosen to be the ball of $\H_0$ centered at
zero of radius $R_0=1+\sup\big\{\|z\|_{\H_0}:\L(z)\leq K\big\}$,
where $K=1+\sup_{z_0\in{\mathcal S}}\L(z_0)$. Note that $R_0$ can
be explicitly calculated in terms of the structural quantities of
our system.


\section {Exponential stability}

\noindent Recalling Theorem \ref{TH-stat.solut}, the set
$\mathcal{S}_0$ of stationary solutions reduces to a singleton
when
\begin{equation}
\label{beta} \beta \geq -\beta_c(k)= -\min_{n\in\N}\mu_n(k),\qquad
\mu_n(k)= \sqrt{\lambda_n}\left[1+\frac{k}{\lambda_n}\right], \ \
\lambda_n= n^4\pi^4.
\end{equation}
It is worth noting that $\beta_c(k)$ is a piecewise-linear
function of $k$, in that $\beta_c(k)=\mu_1(k)$ when
$0<k<\sqrt{\lambda_{1}} \sqrt{\lambda_{2}}=4\pi^4$, and in general
$$\beta_c(k)=\mu_n(k)\quad  \hbox{when}\quad  \sqrt{\lambda_{n-1}}\sqrt{\lambda_{n}}<k< \sqrt{\lambda_{n}}\sqrt{\lambda_{n+1}}.$$
Unlike the case $k=0$, the energy $\E(t)$ does not decay
exponentially in the whole domain of the $(\beta,k)$ plain where
\eqref{beta} is satisfied, but in a region which is strictly
included in it.

More precisely, let
\begin{equation} \nonumber
\bar{\beta}(k)= \left\{
\begin{array}{ll}
\beta_{c}(k), \quad \quad \quad \qquad 0<k\leq \lambda_1,
\\[1em]
2\sqrt{k}, \quad \quad \quad \qquad k>\lambda_1,
\end{array}
\right.
\end{equation}
the following result holds
\begin{theorem}
\label{exp-stab} When $f\,=0$, the solutions to \eqref {FIRST}
decay exponentially, i.e.
$$
\E(t)\,\leq\,c_0\,\E(0)\,e^{- c t}
$$
with $c_0$ and $c$ suitable positive constants, if and only if
$\beta\,>-\bar{\beta}(k)$.
\end{theorem}
Using the same strategy bolstered in \cite{BGV1,BOV}, the proof of
this theorem is a direct consequence of the following lemma:
\begin{lemma}{\rm{(see \cite{BOV})}}
\label{lemma} Let $\beta \in \mathbb{R}$, $k>0$ and
$$
Lu\,=Au\,+\,\beta A^{\frac{1}{2}}u\,+k\,u \ .
$$
There exists a real function $\nu\,=\,\nu(\beta,k)$ such that
$$\left\langle Lu\,,\,u\right\rangle \geq \nu \left\| u\right\| _{2}^{2},$$
where $\nu(\beta,k)>0$ if and only if $\beta >-\bar{\beta}(k)$.
\end{lemma}

\begin{remark}
We stress that Theorem \ref{exp-stab} holds even if $k=0$. In this
case however we have $\bar{\beta}(0)= - {\beta_{c}}(0)= -
\sqrt{\lambda_{1}}$. Then the null solution is exponentially
stable, if unique.
\end{remark}


\section{The Global Attractor}

\noindent Now we state the existence of a global attractor for
$S(t)$, for any $\beta \in \mathbb{R}$ and $k \geq 0$. We recall
that the global attractor $\A$ is the unique compact subset of
$\H_0$ which is at the same time, {\it fully invariant}, i.e.
$S(t)\A=\A,$ for every $t\geq 0$ and {\it attracting}, i.e.
$$\lim_{t\to\infty}\boldsymbol{\delta}(S(t)\B,\A)\to 0,$$
for every bounded set $\B\subset\H_0$, where $\boldsymbol{\delta}$
stands for the Hausdorff semidistance in $\H_0$ (see \cite{BV},
\cite{HAL}, \cite{TEM}).

We shall prove the following
\begin{theorem}
\label{MAIN} The semigroup $S(t)$ on $\H_0$ possesses a connected
global attractor $\A$ bounded in $\H_2$, whose third component is
included in $\D(T)$, bounded in $\M^1_0$ and pointwise bounded in
$H_4$. Moreover, $\A$ coincides with the unstable manifold of the
set ${\mathcal S}$ of the stationary points of $S(t)$, namely,
$$\A=
\left\{ \widetilde{z}\,\,(0):
\begin{array}{cc}
      &  \widetilde{z}\,\,\hbox{is a complete (bounded) trajectory of }S(t):
       \\
 & \underset{t\rightarrow
\infty }{\lim }\left\| \widetilde{z}(-t)-\,S\right\| _{\H_{0}}=\,0
\end{array}
\right\} .
$$
\end{theorem}

\noindent The set ${\mathcal S}$ of all stationary solutions
consists of the vectors of the form $(u,0,0)$, where $u$ is a
(weak) solution to the equation
$$Au+\big(\beta+\|u\|^2_1\big)A^{1/2}u+ku= f.$$
It is then apparent that ${\mathcal S}$ is bounded in $\H_0$. If
${\mathcal S}$ is finite, then
\begin{equation}
\label{CorFinite} \A= \Big\{\tilde z(0): \lim_{t\to
\infty}\|\tilde z(-t)-z_1\|_{\H_0} =\lim_{t\to \infty}\|\tilde
z(t)-z_2\|_{\H_0}=0\Big\},
\end{equation}
for some $z_1,z_2\in{\mathcal S}$.

\begin{remark}
When $f=0$, $k\geq 0$ and $\beta \geq -\beta_c(k)$, then
$\A=\mathcal{S}=\{(0,0,0)\}$. If $\beta < -\beta_c(k)$, then
$\mathcal{S}=\mathcal{S}_0$ may be finite or infinite, according
to Theorem \ref{TH-stat.solut}. In the former case,
\eqref{CorFinite} applies.
\end{remark}

The existence of a Lyapunov functional, along with the fact that
${\mathcal S}$ is a bounded set, allow us prove the existence of
the attractor exploiting a general result from \cite{CP}, tailored
for our particular case.

\begin{lemma}{\rm{(see \cite{CP})}}
\label{lemmaABSTRACT} Assume that, for every $R>0$, there exist a
positive function $\psi_R$ vanishing at infinity and a compact set
$\K_R\subset\H_0$ such that the semigroup $S(t)$ can be split into
the sum $L(t)+K(t)$, where the one-parameter operators $L(t)$ and
$K(t)$ fulfill
$$\|L(t)z\|_{\H_0}\leq \psi_R(t)\qquad\text{and}\qquad
K(t)z\in\K_R,$$ whenever $\|z\|_{\H_0}\leq R$ and $t\geq 0$. Then,
$S(t)$ possesses a connected global attractor $\A$, which consists
of the unstable manifold of the set ${\mathcal S}$.
\end{lemma}

\noindent The proof of Theorem \ref{MAIN} will be carried out be
showing a suitable asymptotic compactness property of the
semigroup, obtained exploiting a particular decomposition of
$S(t)$ devised in \cite {GPV}.

By the interpolation inequality
$$\|u\|_1^2\leq \|u\|\|u\|_2,$$
it is clear that, provided that $\gamma>0$ is large enough,
\begin{equation}
\label{gamma} \frac12\|u\|_2^2\leq
\|u\|_2^2+\beta\|u\|_1^2+\gamma\|u\|^2\leq m\|u\|_2^2,
\end{equation}
for some $m=m(\beta,\gamma)\geq 1$. Now, choosing $\gamma= \alpha
+k$, where $k>0$ is a fixed value, we assume $\alpha$ large enough
so that \eqref{gamma} holds true. Let $R>0$ be fixed and
$\|z\|_{\H_0}\leq R$. Paralleling the procedure given in
\cite{GPV}, we decompose the solution $S(t)z$ into the sum
$$S(t)z=L(t)z+K(t)z,$$
where
$$L(t)z=(v(t),\pt v(t),\xi^t)\qquad\text{and}\qquad
K(t)z=(w(t),\pt w(t),\zeta^t)$$ solve the systems
\begin{equation}
\label{DECAY}
\begin{cases}
\ptt v+Av+\displaystyle\int_0^\infty \mu(s) A\xi(s)ds+
(\beta+\|u\|^2_1)A^{1/2}v+\alpha v+ kv= 0,\\
\pt \xi=T\xi+\pt v,\\
\noalign{\vskip1.5mm} (v(0),\pt v(0),\xi^0)=z,
\end{cases}
\end{equation}
and
\begin{equation}
\label{CPT}
\begin{cases}
\ptt w+Aw+\displaystyle\int_0^\infty \mu(s) A\zeta(s)ds+
(\beta+\|u\|^2_1)A^{1/2}w-\alpha v + kw= f,\\
\pt \zeta=T\zeta+\pt w,\\
\noalign{\vskip1.5mm} (w(0),\pt w(0),\zeta^0)=0.
\end{cases}
\end{equation}

\noindent Then, Theorem \ref{MAIN} is proved as a consequence of
the following lemmas.

\begin{lemma}
\label{lemmaDECAY} There is $\omega=\omega(R)>0$ such that
$$\|L(t)z\|_{\H_0}\leq Ce^{-\omega t}.$$
\end{lemma}

\noindent It shows the exponential decay of $L(t)z$ by means of a
dissipation integral (see Lemma 5.2 of \cite {GPV}).

\begin{lemma}
\label{lemmaCPT} The estimate
$$\|K(t)z\|_{\H_2}\leq C$$
holds for every $t\geq 0$.
\end{lemma}

\noindent It shows the asymptotic smoothing property of ${K}(t)$
in a more regular space, for initial data bounded by $R$ (see
Lemma 6.3 of \cite {GPV}).

The proof of these two lemmas can be done following the guidelines
of \cite {GPV}, remembering that in definition of the coefficient
$\alpha$ in the functional $\Phi_0$ (see \cite{GPV} pag. 726) also
the elastic coefficient $k$ must be taken into account.

\end{document}